\def\full{1}

\ifodd\full%
\newcommand{\ifFull}[1]{#1}
\newcommand{\ifFullElse}[2]{#1}
\else
\newcommand{\ifFull}[1]{}
\newcommand{\ifFullElse}[2]{#2}
\fi

\ifFullElse{
	\documentclass{article}
	\usepackage[letterpaper,top=2cm,bottom=2cm,left=2.5cm,right=2.5cm,marginparwidth=1.75cm]{geometry}
}
{
	\documentclass{llncs}
}

\ifFull{
	\usepackage{amsthm}
}

\usepackage{placeins}
\usepackage{amsmath}
\usepackage{amssymb}
\usepackage[utf8]{inputenc}
\usepackage{alphabeta}
\usepackage{amssymb}
\usepackage{pgf, tikz}
\usetikzlibrary{arrows, automata, arrows.meta, shapes,positioning,calc}
\usepackage{hyperref}
\usepackage[utf8]{inputenc}
\usepackage{multirow}
\usepackage{diagbox}
\usepackage{afterpage}
\usepackage{wrapfig}
\usepackage{import}
\usepackage{subcaption}
\usepackage{comment}

\ifFull{
	\newtheorem{theorem}{Theorem}
	\newtheorem{lemma}{Lemma}
	\newtheorem{definition}{Definition}
	\newtheorem{corollary}{Corollary}
	\newtheorem{claim}[theorem]{Claim}
	\newtheorem{remark}[theorem]{Remark}
}

\newtheorem{assumption}{Assumption}
\newenvironment{claimproof}{
		\begin{proof}[of Claim.]
	}{
		\qed
		\end{proof}
	}

\newcommand{\oned}{M^{1}}
\newcommand{\md}{M^{m}}

\newcommand{\ignore}[1]{}

\begin{document}

\title{One-dimensional vs. Multi-dimensional Pricing \\ in Blockchain Protocols}

\ifFullElse{ 
		\author{
		Aggelos Kiayias \\
		University of Edinburgh \& IOG\\ 
		\small{\tt{akiayias@inf.ed.ac.uk}}
		\and
		$\mbox{Elias Koutsoupias}$\\
		University of Oxford\\ 
		\small{\tt{elias@cs.ox.ac.uk}}	
		\and
		$\mbox{Giorgos Panagiotakos}$\\
		IOG	\\
		\small{\tt{giorgos.panagiotakos@iohk.io}}
		\and
		Kyriaki Zioga\\
		National Technical University of Athens\\
		\small{\tt kyriaki.zioga@gmail.com}	
	}	
}
{
		\author{
			Aggelos Kiayias\inst{1} \and
			Elias Koutsoupias\inst{2} \and
			Giorgos Panagiotakos\inst{3} \and
			Kyriaki Zioga\inst{4}
		}
		
	\institute{
		University of Edinburgh \& IOG, UK,
		\email{akiayias@inf.ed.ac.uk}
		\and
		University of Oxford,
		\email{elias@cs.ox.ac.uk}
		\and	
		IOG, Greece, 
		\email{giorgos.panagiotakos@iohk.io}
		\and 
		National Technical University of Athens,
		\email{kyriaki.zioga@gmail.com}
	}
}

\maketitle

\begin{abstract}
 
  Blockchain transactions consume diverse resources, foremost among
  them storage, but also computation, communication, and
  others. Efficiently charging for these resources is crucial for
  effective system resource allocation and long-term economic
  viability. The prevailing approach, one-dimensional pricing, sets a
  single price for a linear combination of resources. However, this
  often leads to under-utilization when resource capacities are
  limited. Multi-dimensional pricing, which independently prices each
  resource, offers an alternative but presents challenges in price
  discovery. 
  
  This work focuses on the welfare achieved by these two
  schemes. We prove that multi-dimensional pricing is superior under
  stable blockchain conditions. Conversely, we show that one-dimensional
  pricing outperforms its multi-dimensional counterpart in transient
  states, exhibiting faster convergence and greater computational
  tractability. These results highlight a critical trade-off: while
  multi-dimensional pricing offers efficiency gains at equilibrium,
  its implementation incurs costs associated with system
  transitions. Our findings underscore the necessity for a deeper
  theoretical understanding of these transient effects before
  widespread adoption. Finally, we propose mechanisms that aim to mitigate
  some of these issues, paving the way for future research.
 
\end{abstract}


\section{Introduction}

Blockchain ledgers, by their very nature, operate under strict resource constraints. The decentralized and replicated design of these systems means every transaction included in a block must be processed, validated, and 
stored by all participating nodes. This distributed overhead, encompassing \emph{storage, computation, and communication}, inherently limits the number of transactions that can be effectively processed within a given timeframe. Consequently, blockchain protocols impose restrictions on how blocks are assembled, leading to a setting where \emph{only a fraction of pending transactions can be admitted into the ledger at any given time}.

To ensure efficient resource allocation and prioritize the most economically valuable transactions, they must be suitably priced to offset the costs incurred by their processing and inclusion. In simpler blockchain architectures, such as Bitcoin, this pricing mechanism is largely \emph{single-resource based}. Transaction inclusion is primarily determined by a ``feerate'', defined as the fee paid per byte (or virtual byte) of transaction data.
Transactions offering a higher feerate are preferentially included, as Bitcoin blocks have an upper bound on ``virtual bytes''. This effectively treats storage as the primary, if not sole, bottleneck resource.

However, as blockchain platforms evolve to support more complex functionalities, the single-resource pricing model becomes increasingly inadequate. Blockchains offering extensive smart contract capabilities, like Ethereum, must account for a wider array of resources beyond mere data storage, including computational steps and memory access. Ethereum, for instance, introduced ``gas'' as a global unit of account, pricing all operations of its virtual machine in terms of the total gas they consume.  While a step towards recognizing diverse resource demands, this still aggregates all resource consumption into a single ``gas'' dimension.

This prevailing approach, which sets a single price for a linear combination of resources, can be characterized as ``\emph{one-dimensional pricing}''. Yet, it is increasingly evident that transactions on sophisticated virtual machines may have distinct and non-linearly related requirements across multiple resource types, such as storage, computational steps, and network bandwidth. This inherent multi-dimensionality of resource consumption naturally suggests that a ``\emph{multi-dimensional}'' fee approach, where each resource is independently priced, might be more appropriate. Recent developments in systems like Ethereum, such as EIP-4844, further underscore this shift by introducing new distinct resource types, like ``blobs'', in addition to the existing ``gas'' mechanism. Blobs are designed for temporary, high-throughput data storage, distinct from the more permanent storage associated with typical transaction data, highlighting the growing recognition of different resource types with unique pricing considerations.

The intuitive appeal of multi-dimensional pricing lies in its potential for more efficient resource allocation, allowing for finer-grained control over demand for specific, constrained resources. However, despite these theoretical advantages, one-dimensional pricing remains the dominant paradigm in current blockchain implementations. This discrepancy points to a critical trade-off that remains not fully understood, despite the recent expansion of literature focusing on improving multi-dimensional pricing~\cite{diamandis2023designing,angeris2024multidimensional,crapis2024optimal,heimbach2025early,lavee2025does}.

\medskip 

This paper focuses on this fundamental trade-off, comparing the welfare achieved by one-dimensional and multi-dimensional pricing schemes for blockchain transactions. We demonstrate that, under stable blockchain conditions and at equilibrium, multi-dimensional pricing theoretically outperforms its one-dimensional counterpart, leading to superior resource utilization and overall welfare. Conversely, we show that in transient states, one-dimensional pricing exhibits faster convergence and greater computational tractability, making it more robust during periods of rapid change. Furthermore, fast convergence may lead to higher welfare as it can reduce the loss of welfare during the suboptimal pricing of the transient phase. These findings highlight a crucial dilemma: while multi-dimensional pricing offers clear efficiency gains at equilibrium, its widespread adoption necessitates a deeper theoretical understanding of the costs associated with system transitions and the challenges of price discovery in a dynamic environment.

\medskip

\noindent
\textbf{Our Results.}
We conduct a comparative investigation of one-dimensional and multi-dimensional transaction fee systems, highlighting their performance in welfare optimization and dynamic behavior.

\begin{itemize}
	\item 
\emph{Welfare Optimization.}
We prove that under stable blockchain conditions and at equilibrium, multi-dimensional pricing achieves superior welfare compared to its one-dimensional counterpart (Theorem~\ref{thm:welfare}). This rigorous finding stems from its ability to more efficiently accommodate 
resource consumption within the same space.

\item \emph{Stabilization Time.}
Conversely, we demonstrate that in transient states, one-dimensional pricing significantly outperforms multi-dimensional pricing in terms of convergence. Specifically, we quantify its faster convergence to the steady state (Theorem~\ref{thm:ratio_convergence}).

\item \emph{Computational Tractability.}
In the context of transient states, we also show that one-dimensional pricing offers greater computational tractability for block producers
when the system is under congestion. We establish this by demonstrating that the relevant revenue maximization problem in that case  is considerably more complex under multi-dimensional pricing, reducing a multi-dimensional variant of the knapsack problem to highlight its (fine-grained) intractability (Theorem~\ref{thm:values}).
\end{itemize}

Our work thus rigorously verifies and quantifies key tradeoffs between these pricing schemes. These results underscore the need for further research into alleviating these tradeoffs, specifically through improved price discovery algorithms for multi-dimensional systems or more adaptive weighting mechanisms for one-dimensional schemes (as discussed further in the Conclusions section).

\paragraph{Related work.} In~\cite{diamandis2023designing}, the authors describe a novel price-update procedure
for the multidimensional mechanism, and show that it maximizes a utility function that takes in account both
the users and the miners as well as other network/protocol related objectives.
Further, in a follow-up~\cite{angeris2024multidimensional}, they show that this price-update procedure
achieves similar performance to choosing optimal fixed prices offline, i.e., after seeing future traffic.
Both works focus on analyzing \emph{only}  multidimensional 
mechanisms. In \cite{crapis2024optimal}, the authors develop a \emph{practical} framework for dynamically pricing multiple blockchain 
resources using linear-quadratic control theory. 
Further, \cite{heimbach2025early}, provides a statistical analysis of Ethereum's multidimensional fee market since EIP-4844's implementation, i.e.,
the introduction of a second synthetic resource, blobs.

The concurrent work of \emph{Lavee et al.}~\cite{lavee2025does} is the one that is most related to ours.
In this paper, the authors study approximately how much more throughput is needed by the one-dimensional
mechanism in order to be able to accommodate any allocation that is feasible in the multi-dimensional mechanism.
They show that calculating the throughput lost by the single-dimension mechanism, is equal to the value of a specific zero-sum game.
Further,  they study multidimensional mechanisms where the
number of dimensions is less than the total number of physical resources.
It is shown that computing the optimal number of dimensions in that case that ensures a specific upper-bound 
on throughput loss is computationally hard.
The authors seem to arrive to similar conclusions as our work, i.e., that is handling 
multi-dimensional constrains introduces computational difficulties.
However, unlike our work, they refrain from exploring welfare or transient states of 
the two mechanisms.

\paragraph{Organization.} In Section~\ref{sec:prelim} we formally introduce the single and multi-dimension mechanisms,
Section~\ref{sec:welfare} explores the welfare generated by the two mechanisms in the stable state.
Next, we turn our attention to transient states.
In Section~\ref{sec:stability}, we study the time to stability of the two mechanisms,
while in Section~\ref{sec:revenue} we identify a computational bottleneck of the multidimensional mechanism
when the system is under congestion.
Finally, in Section~\ref{sec:proposals}, two approaches for side-stepping the limitations outlined above are explored.
\ifFullElse{}{Due to space constraints, some of the proofs are presented in full in the Appendix.}

\vspace{0.5cm}


\section{Preliminaries}\label{sec:prelim}

\ignore{
\[
\begin{array}{|c|c|}
	\hline
	\textbf{Variable} & \textbf{Meaning} \\
	\hline
	b_j & \text{bid of transaction } j\\
	c^j = (c_1^j, ...,c_m^j) & \text{consumption vector of transaction }j \\
	g_j & \text{gas consumption of transaction }j \\
	B_t & \text{block at time t} \\
	\oned & \text{1-dimensional pricing mechanism} \\
	\md   & m\text{-dimensional pricing mechanism} \\
	G & \text{upper bound on gas consumption} \\
	T & \text{target} \\
	r & \text{base fee for gas} \\

	G_i & \text{upper bound of resource }i \\
	T_i & \text{target of resource }i \\
	r_i & \text{base fee for resource i} \\
	W_t^1 & \text{welfare of the one-dimensional pricing mechanism at time t}\\
	W_t^m & \text{welfare of the m-dimensional pricing mechanism at time t}\\
	\hline
\end{array}
\]
}


\subsection{Blockchain protocols and TFMs}
We are going to analyze transaction fee mechanisms (TFMs) run on top of a protocol that maintains a blockchain. Our formalization of TFMs  closely follows that of~\cite{roughgarden2024transaction}.

Each block $B_i$ in a chain of blocks $(B_i)_i$ consists of a list of transactions. 
Each transaction $tx_j$ is associated with: (i) a resource consumption vector $c^j = (c_1^j, c_2^j,\ldots, c_m^j)$, where  $c_i^j$ represents the consumption of resource $i$ by the transaction, and $m$ denotes the total number of resources consumed by the system, e.g., cpu or memory units, 
and (ii) the (private) value $v_j$ of $tx_j$ for the user submitting as well as a bid $b_j$ representing the user's willingness to pay.


The first thing any blockchain protocol has to specify is an allocation rule
that takes as input a blockchain and a set of transactions,
and decides a subset of them--the ones that will be included
in the next block.
Candidate transactions for inclusion are maintained by the block producer in
the \emph{mempool}, 
denoted here by $in$. 

\begin{definition}[Allocation rule]
An allocation rule is a vector-valued function $x$ from the on-chain history $(B_i)_{i \in [k-1]}$ and mempool $in_k$ to a 0-1 value $x_j((B_i)_{i \in [k-1]},in_k)$ for each pending transaction $tx_j \in in_k$. 
\end{definition}
A value of 1 for $x_j((B_i)_{i \in [k-1]}, in_k)$ indicates transaction $tx_j$’s inclusion in the current block $B_k$; a value of 0 indicates its exclusion.
Block producers can choose a different subset of transactions to include in their block than the ones produced by $x$ reflecting their strategy.



For the specific TFMs we consider in this work, the protocol has to further decide how much the creators of transactions included in the blockchain have to pay, and to whom their payment is directed. These decisions are formalized by two more functions: a payment rule $p$ and a burning rule $q$. 

\begin{definition}[Payment rule]
A payment rule is a function $p$ from the current on-chain history $(B_i)_{i \in [k-1]}$ and transactions $B_k$ included in the current block to a nonnegative number $p_j((B_i)_{i \in [k]})$ for each included transaction $tx_j \in B_k$.
\end{definition}
The value of $p_j((B_i)_{i \in [k]})$ indicates the payment from the creator of an included transaction $tx_j \in B_k$ to the miner of the block $B_k$.

\begin{definition}[Burning rule]
A burning rule is a function $q$ from the current on-chain history $(B_i)_{i \in [k-1]}$ and transactions $B_k$ included in the current block to a nonnegative number $q_j((B_i)_{i \in [k]})$ for each included transaction $tx_j \in B_k$.
\end{definition}
The value of $q_j((B_i)_{i \in [k]})$ indicates the amount  burned by the creator of an included transaction $tx_j \in B_k$.

Given $x$, $p$, and $q$, the process of chain-formation  is as follows:
\begin{itemize}
    \item The system starts given some initial block $B_0$ -- the Genesis block;
    \item Block producers generate the next block $B_i$  in the chain $(B_i)_{i\in[i-1]}$ through either functions $(x,p,q)$
    or some other algorithms of their choice.
\end{itemize}

In order to determine whether a blockchain is valid or not, the protocol also specifies a validity function $y$. 
We assume $y$ always respects correctness, i.e, that valid blockchains extended by $x$, $p$, and $q$ remain valid.

\paragraph{Resource bounds.} Finally, for safety reasons, it is essential to enforce strict upper limits on resource consumption within the system. Exceeding these limits can lead to system failures, inefficiencies, increased latency, or even denial-of-service (DoS) scenarios. By maintaining these limits, we ensure that the system remains responsive and resistant to attacks or unintended excessive resource consumption. 

\begin{assumption}[Safety]\label{ass:safety}
	For each $j\in[m]$ and known $(G_j)_{j\in[m]}$, the maximum consumption of resource $j$  in a block 
	is at most $G_j$.
\end{assumption}

We assume that $y$ ensures that valid blocks always satisfy the safety assumption.


\subsection{The one-dimensional mechanism} \label{sec:1d}
We denote by $\oned := (x^1,p^1,q^1,y^1)$ the \emph{one-dimensional pricing mechanism}, similar to the one implemented in Ethereum's EIP-1559~\cite{buterinethereum}. In this mechanism, each block specifies a target consumption $T$ and a maximum consumption $G$.  The target represents the ideal consumption of \emph{gas}, an artificial resource constructed as a linear combination of physical resources.
The upper bound represents the maximum allowed gas consumption (c.f., safety assumption). Additionally, the target is set as half of the upper bound, i.e., $T  :=  G/2$.

Each block has a protocol-computed reserve price (per unit of gas) called the \emph{base fee}. Paying the base fee is a prerequisite for inclusion in a block. The base fee is a function of the preceding blocks only and does not depend on the transactions included in the current block. Specifically, we denote by $\alpha((B_i)_{i \in [k-1]})$, the base fee for the next block, which is determined by a given sequence of preceding blocks $(B_i)_{i \in [k-1]}$.
In particular, the specific adjustment rule - function $α$ - proposed in EIP-1559 computes the base fee $r_{cur}$ for the current block from the base fee $r_{pred}$ and the total gas consumption  $g_{pred}$ of the predecessor block using the following formula:

\begin{equation}\label{eq:eip_update_rule}
    r_{cur} := r_{pred} \cdot \left( 1 + \frac{1}{8} \cdot \frac{g_{pred} - T}{T} \right).
\end{equation}

Note, that the base fee varies according to the total consumption, i.e., when the total consumption exceeds the target, the base fee increases accordingly, and when the total consumption is less than the target, the base fee decreases accordingly.

\begin{definition}[Gas consumption of a transaction in $\oned$] \label{def:gas_consumption_1d}
The gas consumption $g_j$ of a transaction $tx_j$ in the $\oned$ mechanism is defined as a weighted sum of the resource consumption components. Let $(w_i)_{i \in [m]}\in \mathbb{R}^m_{>0}$ be  predefined  weight coefficients representing the different costs of each resource. The gas consumption $g_j$ is computed as:
$$g_j  :=  \sum_{i \in [m]} w_i \cdot c_i^j$$

\end{definition}

A direct consequence of the safety assumption for $M_1$ is the  following useful corollary.
\begin{corollary}\label{cor:safe}
	The total gas consumption must not exceed the upper bound of any resource, i.e., $G \leq w_i \cdot G_i,$ for any $i\in [m].$
\end{corollary}

Besides the base fee, users can decide to pay a tip to the block producer 
for including their transaction. Formally, we define the \emph{tip} to be the difference of the bid and the base fee paid by a transaction, i.e.,  for transaction $tx_j$: $$tip_j := b_j - r\cdot g_j$$
The allocation rule $x^1$ of the $\oned$ mechanism is to include a subset of outstanding transactions that maximizes the sum of tips.

\begin{definition}[Allocation rule]\label{def:m1_allocation_rule}
For each history $(B_i)_{i \in [k-1]}$ and corresponding base fee $r  :=  \alpha((B_i)_{i \in [k-1]})$, 
 the $x^1$’s are assigned 0-1 values to maximize
$$\sum_{tx_{j} \in in_k}x^1_j((B_i)_{i \in [k-1]}, in_k) \cdot 
tip_j$$
subject to the gas constraint in Corollary~\ref{cor:safe}.
\end{definition}

In turn, the payment rule dictates that the tip must be transferred to the block producer, while the burning rule dictates that the base fee must be burned. 

\begin{definition}[Payment Rule] \label{def:m1_payment_rule}
For all $(B_i)_{i\in [k]}$ and $tx_j\in B_k$: $p^1_j((B_i)_{i\in [k]}) :=  tip_j$.
\end{definition}

\begin{definition}[Burning Rule]\label{def:m1_burning_rule}
For all $(B_i)_{i\in [k]}$ and $tx_j \in B_k$, and  $r  :=  α((B_i)_{i \in [k-1]})$: $q^1_j((B_i)_{i\in [k]})  :=  r \cdot g_j$.
\end{definition}

Finally, $y^1$ deems a blockchain $(B_i)_{i \in [k]}$ valid if every transaction in each block has paid its base fee and every block $B$ in the sequence satisfies the gas consumption constraint, i.e., that
\begin{equation}\label{eq:block_size_constraint}
	\sum_{{tx_j} \in B} g_j \leq G 
\end{equation}

\subsection{The multi-dimensional mechanism}

We denote by $\md := (x^m,p^m,q^m,y^m)$ the multi-dimensional pricing mechanism with $m$ dimensions~\footnote{The description of the mechanism is loosely based on public discussions in the Ethereum community~\cite{buterinmulti,EIP-4844}.}. Each block has a target $T_i$ and an upper bound $G_i$ for each resource; $G_i$ is  the upper bound stated in safety assumption. The target is set as half of the upper bound, i.e., $T_i  := G_i/2$. Each resource is priced independently of the others according to its demand. Therefore, each resource $i$ has a corresponding base fee $r_i$, for a total base fee of transaction $tx_j$ of $\sum_{i\in[m]} r_ic^j_i$ and corresponding tip $tip_j := b_j - \sum_{i\in[m]} r_ic^j_i$.

As in the one-dimensional mechanism, paying the base fees is a prerequisite for inclusion in a block. The base fees are a function of the preceding blocks only and do not depend on the transactions included in the current block. For simplicity, we will assume that the individual resource base fees are determined in the same way as the gas base fee.
We denote the corresponding functions by $a_i$, with the price update equation for resource $i$ being:~\footnote{While in reality details of the price update rule may be different for different resources, for simplicity we here we assume that it is the same for all of them.}
\begin{equation}\label{eq:md_update_rule}
	r_{i,cur} := r_{i,pred} \cdot \left( 1 + \frac{1}{8} \cdot \frac{\sum_{tx_j \in B_{pred}} c_i^j - T_i}{T_i} \right)
\end{equation}

The (intended) allocation rule $x^m$ of the $\md$ mechanism is the same as that of $\oned$, i.e.,  include a subset of outstanding transactions that maximizes the sum of the tips

\begin{definition}[Allocation rule]\label{def:md_allocation_rule}
For each history $(B_i)_{i \in [k-1]}$ and corresponding base fees $r_i := \alpha_i((B_i)_{i \in [k-1]})$ for all $i \in [m]$, the $x^m$’s are assigned 0-1 values to maximize
$$\sum_{tx_j \in in_k}x^m_j((B_i)_{i \in [k-1]}, in_k) \cdot tip_j$$
subject to the block size constraints of safety assumption.
\end{definition}

The payment and burning rules are defined similarly to that of $\oned$.
\begin{definition}[Payment Rule] \label{def:md_payment_rule}
For all $(B_i)_{i\in [k]}$ and $tx_j\in B_k$: $p^m_j((B_i)_{i\in [k]}) := tip_j$
\end{definition}

\begin{definition}[Burning Rule]\label{def:md_burning_rule}
For all $(B_i)_{i\in [k]}$ and $tx_j \in B_k$, letting $r_i  :=  \alpha_i((B_i)_{i \in [k-1]})$ for all $i \in [m]$: $q^m_j((B_i)_{i\in [k]}) := \sum_{i \in [m]}r_i \cdot c_i^j$.
\end{definition}

The predicate $y^m$ deems a blockchain $(B_i)_{i \in [k]}$ as being valid if every transaction in each block has paid its base fee and every block $B$ in this sequence satisfies the constraints of safety assumption, i.e.,
\begin{equation}\label{eq:block_size_constraint_md}
	\sum_{tx_j \in B} c_i^j \leq G_i \quad \forall i \in [m]
\end{equation}



\section{Welfare Comparison in the Stable State}\label{sec:welfare}

In this section, we compare the welfare of blocks generated by $\oned$ and $\md$, when prices are stable and the two mechanisms see the same set of transactions in the mempool. We show that the welfare generated by $\md$ is
always more than that by $\oned$.

We start with some preliminary definitions.
\emph{Welfare} is defined as the sum of the values of all transactions added to the block, i.e., $\sum_{tx_{j} \in B} v_j $. 
The price(s) of the two mechanisms are said to be \emph{stable}
if  they remain the same in the next block.
Users are \emph{truthful} if their
bid matches the value of the transaction.~\footnote{In the stable state, it has been shown in~\cite{roughgarden2024transaction} that bidding truthfully is a symmetric ex-post Nash Equilibrium for myopic users in $\oned$. The argument is that since this a reserved price mechanism, and assuming there is enough space in the block, the user stating its value truthfully is always optimal in the short-term. Note, that a similar argument can be made for the multidimensional mechanism, since it also uses reserved prices.}

The main idea behind our result, is to closely examine two sets: the set of 
transactions included only by $\md$ and the set of transactions included only by $\oned$. For the first set, it is straightforward to lower bound welfare by the sum of each resource consumed times each price; due to stability we know \emph{exactly} how much of the resources was consumed.
On the other hand, for the transactions on the second set, we know that they are not included by $\md$, and thus their value (and consequently the welfare they generate)  must be less than the value generated in $\md$ by the consumption of the corresponding amount of resources. Otherwise, they would have been included by $\md$. Combining this observation with the fact
that due to the safety assumption the  consumption in $\md$ is larger than in $\oned$ allows us to prove our result.


\begin{theorem}\label{thm:welfare}
	The welfare generated by a block produced by $\md$ is larger than that by $\oned$, given that: 
	(i) prices in both mechanisms are stable, 
	(ii) block producers see the same set of transactions in their mempools,
	(iii) users truthfully report their values, and
	(iv) safety assumption holds. 
\end{theorem}


\begin{proof}
For the sake of contradiction assume that $B^1$, the block generated by $\oned$, has higher welfare than $B^m$, the block generated by $\md$.
First, note that price stability implies that consumption is equal to the target, i.e., in $\oned$
$$\sum_{tx_j \in B} \sum_{i \in [m]} w_i c^j_i = T$$
and in $\md$, for all $i \in  [m]$ it holds that:
$$\sum_{tx_j \in B} \sum_{i \in [m]} c^j_i = T_i.$$

Now, let $S = B^1 \cap B^m$ be the transactions included in both systems. Define the residual targets after excluding $S$:
$$T' = T - \sum_{tx_j \in S}\sum_{i \in [m]}w_i c_i^j, \quad T_i' = T_i - \sum_{tx_j \in S}c_i^j \text{ for all }i \in [m]$$

Note that the set $S$ may coincide with $B^1$, but it must be a proper subset of $B^m$. This is because in the $\oned$, the total gas consumption is exactly $T$, which by safety assumption is less than or equal to each individual term $w_i T_i$, i.e., $\sum_{i \in [m]} w_i \sum_{tx_j \in B^1}c_i^j \leq w_i T_i$ for all $i \in [m]$. This implies that in $\oned$, for at most one resource $k$ the total consumption can be equal to $T_k$. For the sake of contradiction, assume that this condition is met for two resources, say $k$ and one more.
Then we have:
$$T = \sum_{tx_j \in B^1}\sum_{i \in [m]}w_ic_i^j = w_kT_k + \sum_{tx_j \in B^1} \sum_{i \in [m]: i \neq k}w_ic_i^j$$
But since $T$ is less than or equal to $w_kT_k$, it must be that: 
$$\sum_{tx_j \in B^1} \sum_{i \in [m]: i \neq k}w_ic_i^j = 0$$
which is a contradiction to our initial claim.
In contrast, in  $\md$, due to stability all the total resource consumption is going to be higher, and  thus  $B^m$ must include additional transactions beyond those in $S$. 
 
Next, we show a lower bound on the welfare generated by the transactions in $B^m\setminus S$.
If a transaction is only contained in $\md$, i.e., it belongs to the set $B^m \setminus S$, due to truthfulness it must be the case that its value is at least equal to the base fee, i.e., $v_j \geq \sum_{i \in [m]}r_i c_i^j$. Otherwise,  it  would not have been selected by the mechanism. Note, that all transactions that can pay the base fees are included by $\md$, as the system is assumed to be in stability, and thus there is available space to include all transactions paying the required fee.
By summing over all transactions in $B^m \setminus S$, we obtain the desired lower bound on welfare: 
\begin{equation}\label{eq:wel2}
\sum_{tx_j \in B^m \setminus S} v_j \geq \sum_{tx_j \in B^m \setminus S} \sum_{i \in [m]}r_i c_i^j = \sum_{i \in [m]}r_i \sum_{tx_j \in B^m \setminus S} c_i^j = \sum_{i \in [m]} r_i T_i'
\end{equation}
where the first equality follows by switching the order of summation and the second one  from the inequality below:
$$\sum_{tx_j \in B^m \setminus S} c_i^j = \sum_{tx_j \in B^m} c_i^j - \sum_{tx_j \in S} c_i^j = T_i - \sum_{tx_j \in S} c_i^j = T_i'$$

Next, we show a similar upper bound on the welfare generated by transactions in $B^1 \setminus S$.
By safety assumption, we have that $T$ is less than or equal to $w_i T_i$, for all $i \in [m]$. Given the targets $T$, $T'$, for all $i \in [m]$ we can get that:
\begin{equation}\label{eq:T'_bound}
T' \leq w_i T_i - \sum_{tx_j \in S}\sum_{i \in [m]}w_i c_i^j\leq w_i T_i - \sum_{tx_j \in S} w_i c_i^j = w_i (T_i - \sum_{tx_j \in S} c_i^j) = w_i T_i' 
\end{equation}
where the first inequality follows from the definition of $T'$ and safety assumption. The second inequality holds because the sum over all resources is higher than or equal to the contribution of any single resource. In the first equality, the weight $w_i$ is factored out of the summation, and in the second equality, we apply the definition of $T_i'$.

We know that $T'$ represents the total weighted resource consumption of transactions in $B^1 \setminus S$. For each resource $i$, 
$T'$ can be lower bounded by:
\begin{align*}
	T' &= T - \sum_{tx_j \in S}\sum_{i \in [m]}w_i c_i^j = \sum_{tx_j \in B^1}\sum_{i \in [m]}w_i c_i^j - \sum_{tx_j \in S}\sum_{i \in [m]}w_i c_i^j 
	\\
	&= 
	\sum_{tx_j \in B^1 \setminus S} \sum_{i \in [m]} w_i c_i^j \geq \sum_{tx_j \in B^1 \setminus S} w_i c_i^j
\end{align*}

Combining this with Inequality~\ref{eq:T'_bound}, and dividing by $w_i$,  we get that the total consumption of each resource from the transactions that are included only in the $\oned$ is less than or equal to the total consumption of the respective resource in the $\md$:
\begin{equation}\label{eq:bound_1d}
\sum_{tx_j \in B^1 \setminus S} c_i^j \leq T_i', \quad \text{for all } i \in [m]
\end{equation}

If a transaction is only contained in $\oned$, i.e., it belongs to the set $B^1 \setminus S$, due to truthfulness it is implied that it has a value that is less than what it would have to pay to be included in $\md$, i.e., $v_j < \sum_{i \in [m]} r_i c_i^j$. Otherwise, it must have been included in $\md$, which leads to a contradiction.
Summing up we get the following:
\begin{equation}\label{eq:wel1}
\sum_{tx_j \in B^1 \setminus S} v_j < \sum_{tx_j \in B^1 \setminus S} \sum_{i \in [m]} r_i c_i^j = \sum_{i \in [m]} r_i \sum_{tx_j \in B^1 \setminus S} c_i^j \leq 
\sum_{i \in [m]} r_i T_i'  
\end{equation}
where the equality follows by switching the order of summation and the last inequality follows from inequality \ref{eq:bound_1d}.

Combining the Inequalities ~\ref{eq:wel2} and~\ref{eq:wel1}, it follows that $\sum_{tx_j \in B^1 \setminus S} v_j$ is less than $\sum_{tx_j \in B^m \setminus S} v_j$.
Adding the common transactions in $S$ to both sides, we get that
$\sum_{tx_j \in B^1 } v_j$ is less than $\sum_{tx_j \in B^m } v_j$, 
contradicting our initial hypothesis about the welfare generated in $B^1$ and $B^m$.
\qed
\end{proof}

A related question to the one we answered above, is whether the welfare gap can be upper bounded.
A first answer is \emph{no}, since there may exist transactions of arbitrarily high value that if included in $\md$ lead to stability, while if included in $\oned$ make the system unstable. To see this, 
take a transaction that consumes an amount of resources that is equal to the
resource targets $T_i$ for all $i\in[m]$. Even if this transaction is admissible by $\oned$, 
it will not lead to a stable price as its gas consumption can be shown to exceed $T$. 

In reality, the answer is slightly more complicated, as in many deployed systems~\footnote{E.g.,~\cite{EIP-7825}.} there is an upper bound on the amount of resources a transaction may consume that is a fraction of the $T_i$'s.

\medskip

We have shown that $\md$ outperforms $\oned$ in terms of welfare generation in the stable state. Next, we focus on comparing the two mechanisms in different transient states.

\section{Stabilization time}\label{sec:stability}

Next, we examine the behavior of the two mechanisms at first transient-state of interest: the system being in the process of getting \emph{stable}.
We focus in the time to stability after a demand shock, and show that as long as the individual price stabilization time distributions have a non-zero tail, $\oned$  outperforms $\md$.

We start by describing the experiment we are going to use to access the stabilization time of the two  mechanisms. The experiment concerns a blockchain protocol where blocks are produced using one of the two mechanisms of interest, either $\oned$ or $\md$. Initially, mempools are assumed to be filled according to some distribution ${\cal F}$ and the system to have run long enough for prices to be stable. Then, at some time $t$, an unexpected event occurs, causing the system to destabilize. This is captured by assuming that the mempools after $t$ are filled according to a different distribution ${\cal F}'$. Given this setup, we define the \emph{stabilization time} as the time required for the price(s) to be stable again.

\medskip

Next, we introduce some notation for our analysis. By random variable $Z$ we denote the time required for the price to be stable again when our experiment is run with $\oned$, while by $Z^m$ in the case of $\md$. 
Further, let random variable $Z_i$ denote the time it takes for the price of resource $i$ to be stable, for $i\in [m]$, in the multidimensional case. The quantity we will focus on analyzing is the expectation ratio between $Z$ and $Z^m$, i.e., 
$$\text{Ratio}_m := \mathbb{E}[ Z^m ] / \mathbb{E}[Z] $$
$\text{Ratio}_m$ increasing rapidly with $m$, suggests that the mechanism becomes impractical for multi-resource systems due to excessive stabilization time. Conversely, if the ratio is negligible, then the mechanism remains viable even when we increase the number of priced resources.

\medskip 

We are going to make two simplifying assumptions in our analysis.
Our first assumption, has to do with the total consumption distribution of each resource in $\md$. Namely, it says that it only depends on the price of  this resource.

\begin{assumption}\label{ass:D_i}
For the mempool distribution ${\cal F}'$, it holds that for each resource $i \in [m]$, the total consumption of this resource by valid transactions~\footnote{By valid transactions we mean the transactions  in the mempool whose value exceeds the base fee and thus are eligible for inclusion.} in the mempool, denoted by $D_i$, depends only on the related base fee $r_i$.
\end{assumption}

The motivation behind this simplifying assumption is that the price of the resource is
expected to be the main motivating factor for increasing or decreasing its consumption, e.g., $r_i$ going up would most likely lead some users to decreasing its consumption. 

Our second simplification is to assume that the distributions of stabilization times are more or less the \emph{same} for 
all prices (both in $\oned$ and $\md$). The motivation is that prices are adjusted multiplicatively  (Equation~\ref{eq:eip_update_rule}, Section~\ref{sec:1d}),  suggesting that the relevant update procedure approaches the stable price at a similarly fast rate, that is independent of the individual resource considered.
Similarity of distributions is formalized by upper-bounding the closeness of the statistical distance of the relevant distributions by some parameter $\delta$ (which is expected to be small).

\begin{assumption}\label{ass:stabilization}
There exists some $\delta \in [0,1)$, such that for all $i\in[m]$ it holds that $\Delta(Z,Z_i) \leq \delta$.
\end{assumption}

Given these two assumptions, we  can show that the expectation ratio  positively depends on: (i) the size of the  tail~\footnote{If the size of the tail is $0$, i.e., all probability mass is on $\mathbb{E}[Z]$, then it easy to see that $\text{Ratio}_m$ is close to $1$, as all prices stabilize at about the same time.} of $Z$, reflected by parameters $c,p$ in the theorem, and  (ii) the number of dimensions $m$; an error term related to $\delta$ also appears in the ratio.
The intuition is that as the number of dimensions grow, the probability that one of the stabilization times ($Z_i$) falls into the tail of the distribution increases.
We are now ready to formally prove our theorem.

\begin{theorem}[Expectation Ratio] \label{thm:ratio_convergence}
	Given Assumptions \ref{ass:D_i} and \ref{ass:stabilization}, for any constant $c>0$,	it holds that:
	$$\text{Ratio}_m 
	\geq (1 + c) \left( 1 - (1 - p - \delta)^{m-1} \right )  \frac{ p}{p + \delta} - O(\delta/p)$$
	where $p := \mathbb{P}[Z \geq (1+c) \mathbb{E}[Z]]$.
\end{theorem}

\begin{proof}
	Our first observation is that  
	we can lower-bound $Z^m$ by $\max_{i \in [m]} Z_i$, as individual price needs to stabilize before all of them are stable; for the remaining of this proof assume that $Z^m:=\max_{i \in [m]} Z_i$. Similarly, we define $Z^j := \max_{i \in [j]} Z_i$, for all $j\in[m]$.  
	Further, we note that due to Assumption~\ref{ass:D_i}, the $Z_i$'s form a set of mutually independent random variables. To see this, observe that 
	the value of $Z_i$ only depends on the consumption of resource $i$, which in turn only depends on value $r_i$, which is independent of any other value, and thus any other $Z_j$.

	Let now $k:=  (1+c)\mathbb{E}[Z]$.
	We apply the law of total expectation by partitioning the sample space based on whether $Z_m$ is greater than or equal to $k$ or $Z_m$ is less than $k$:
	$$\mathbb{E}[Z^m] = \mathbb{E}[Z^m | Z_m \geq k] \mathbb{P}[Z_m \geq k] + \mathbb{E}[Z^m | Z_m < k] \mathbb{P}[Z_m < k]$$
	
	We then establish bounds for each part of the expectation. Consider the case when $Z_m$ is greater than or equal to $k$. Since the maximum of $Z^m$ must be at least $k$ in this scenario, it follows that $\mathbb{E}[Z^m|Z_m \geq k]$ is greater than or equal to $k$. Furthermore,
	by Assumption~\ref{ass:stabilization} and the properties of statistical distance,
	we get that 
	$$|\mathbb{P}[Z \geq k] - \mathbb{P}[Z_i \geq k]| \leq \delta, \text{ for all } k \in \mathbb{N}$$ 
	which implies that:
	$$\mathbb{P}[Z_m \geq k] \geq \mathbb{P}[Z \geq k] - \delta = p - \delta$$
	
	Now consider the complementary case when $Z_m$ is less than $k$. Adding this ``bad'' sample to the collection of the other $m - 1$ samples cannot reduce the maximum value. Therefore, due to the independence assumption, $\mathbb{E}[Z^m|Z_m < k]$ is greater than or equal to $\mathbb{E}[Z^{m-1}]$. The associated probability is also bounded from below using Assumption~\ref{ass:stabilization}
	as before: 
	$$\mathbb{P}[Z_m < k] \geq \mathbb{P}[Z < k] - \delta = 1 - p - \delta$$
	
	Substituting these bounds into our equation:
	$$\mathbb{E}[Z^m]\geq k(p - \delta) + \mathbb{E}[Z^{m-1}] (1-p - \delta)$$
	
	Dividing both sides by $\mathbb{E}[Z]$ to obtain the ratio:
	$$\text{Ratio}_m = \mathbb{E}[Z^m] / \mathbb{E}[Z] \geq \frac{k(p - \delta) + \mathbb{E}[Z^{m-1}] (1-p - \delta)}{\mathbb{E}[Z]}$$
	
	Since $k = (1+c)\mathbb{E}[Z]$, we have $\frac{k}{\mathbb{E}[Z]} = 1+c$. Therefore:
	$$\text{Ratio}_m \geq (1+c)(p - \delta) + \text{Ratio}_{m-1} (1 - p - \delta)$$
	
	Setting $a = (1+c)(p - \delta)$ and $b = (1 - p - \delta)$, we obtain:
	$$\text{Ratio}_m \geq a + b \cdot \text{Ratio}_{m-1}$$
	
	Further, note that this inequality can be derived for $\text{Ratio}_{j}$ for any $j \in [m]$, by making the same argument starting from the respective $Z^j$. By repeatedly applying this inequality starting from $\text{Ratio}_m$, we get:
	$$\text{Ratio}_m \geq a + b \cdot \text{Ratio}_{m-1} = a(1 + b) + b^2 \text{Ratio}_{m-2} \geq \ldots\geq a\sum_{i=0}^{m-2} b^i + b^{m-1}\text{Ratio}_1$$
	
	Using the formula for the sum of a geometric series: $\sum_{i=0}^{n} r^i = \frac{1 - r^{n+1}}{1 - r}$ for $r \neq 1$ with $r = b$ and $n = m-2$, and applying the lower bound that $\text{Ratio}_1 \geq 0$, it follows that:
	$$\text{Ratio}_m \geq a \frac{1 - b^{m-1}}{1-b}$$
	Restoring the original expressions:
	$$\text{Ratio}_m \geq (1+c)(p - \delta) \frac{1 - (1 - p - \delta)^{m-1}}{p + \delta}$$
	
	\ignore{
	Restoring the original expressions:
	$$\text{Ratio}_m \geq (1+c)(p - \delta)\sum_{i=0}^{m-2} (1 - p - \delta)^i + (1-p-\delta)^{m-1}\text{Ratio}_1$$
	
	Using the formula for the sum of a geometric series: $\sum_{i=0}^{n} r^i = \frac{1 - r^{n+1}}{1 - r}$ for $r \neq 1$ with $r = (1 - p - \delta)$ and $n = m-2$, and applying the lower bound that $\text{Ratio}_1 \geq 0$, it follows that:
	$$\text{Ratio}_m \geq (1+c)(p - \delta) \frac{1 - (1 - p - \delta)^{m-1}}{p + \delta}$$
}	
	The theorem follows.
\end{proof}

As an example, taking $c$ equal to $2$, and assuming the tail of $Z$ goes beyond $3\mathbb{E}[Z]$, $\lim\limits_{m\rightarrow \infty} \text{Ratio}_m$ tends to $(1 + 2) \cdot \frac{p}{p + \delta} \approx 3$ for small $\delta$. This means that the expected stabilization time of $\md$ is going to be approximately three times that of $\oned$. The longer and heavier the tail, the better our bound gets.

Finally, in Figure~\ref{fig:all_distributions}, we show how the expectation ratio  grows with respect to $m$ when $Z$ and $(Z_i)_{i\in [m]}$ follow four well-known distributions: the Geometric, Poisson, Negative Binomial and Logarithmic; again we assume that we are in the best-case for $\md$, where the price of resource $i$ remains stable after it is stabilized once.



\begin{figure}[h!]
	\centering
	
	\begin{subfigure}[b]{0.45\textwidth}
		\centering
		\includegraphics[width=\textwidth]{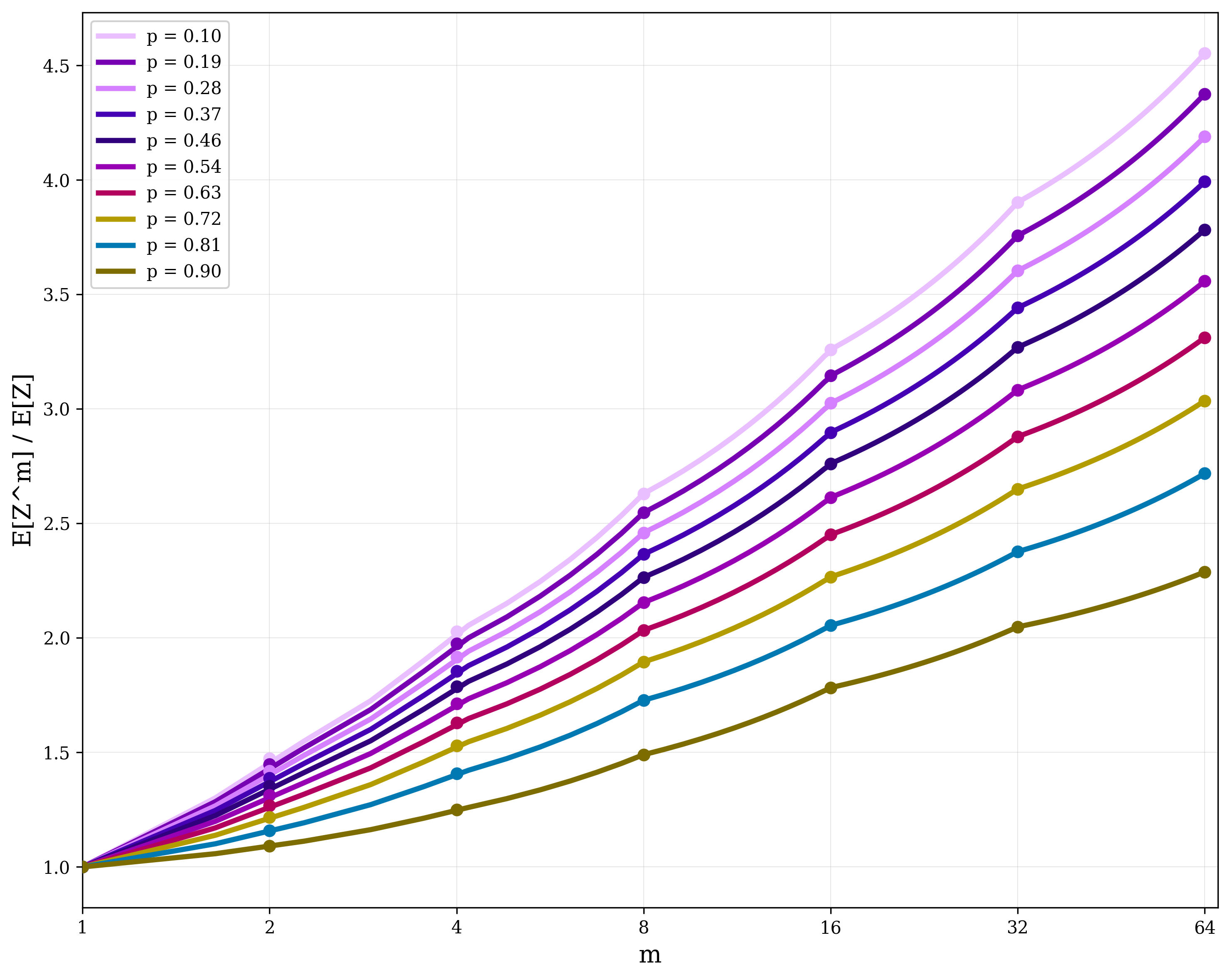}
		\caption{Geometric Distribution}
		\label{fig:geometric}
	\end{subfigure}
	\begin{subfigure}[b]{0.45\textwidth}
		\centering
		\includegraphics[width=\textwidth]{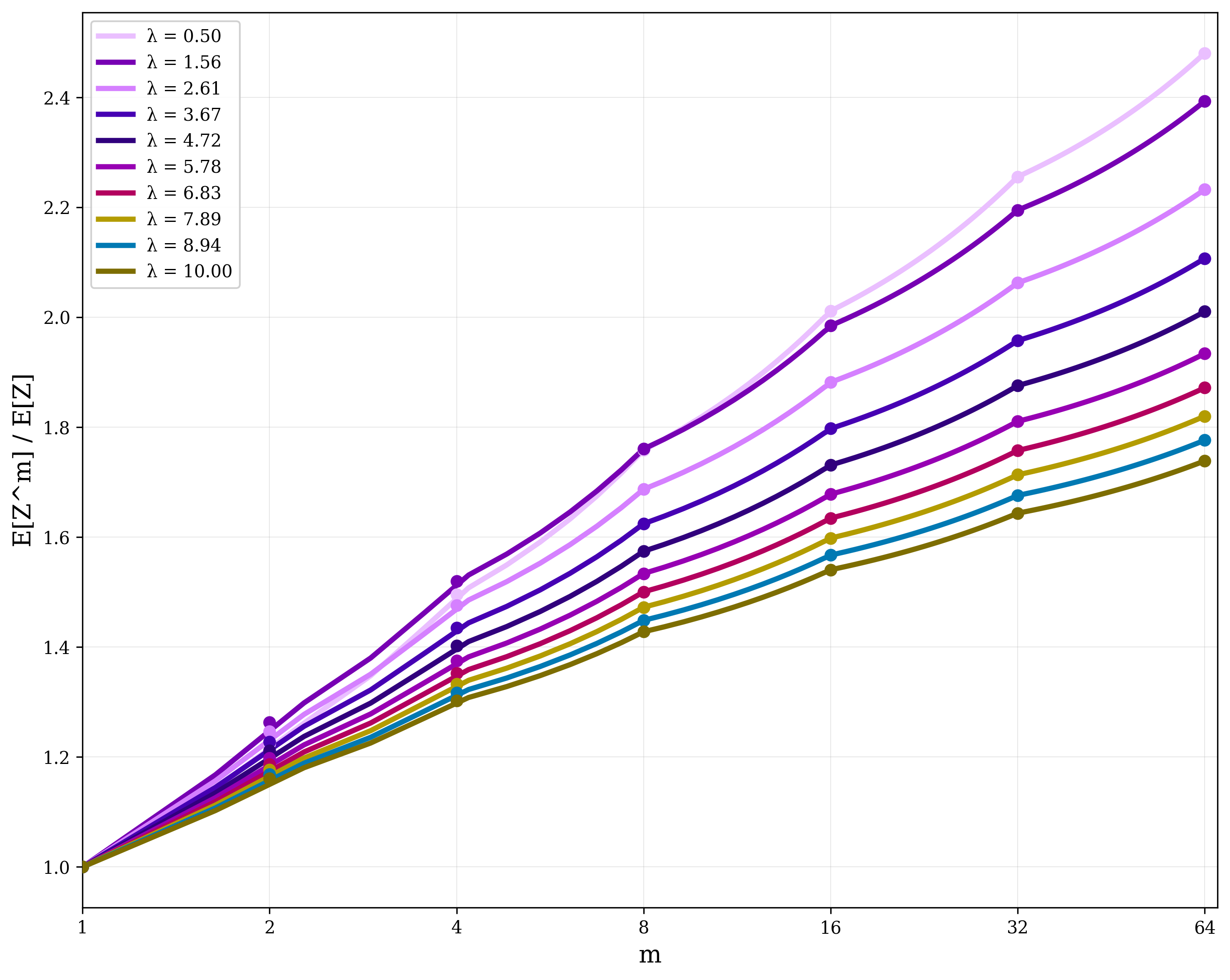}
		\caption{Poisson Distribution}
		\label{fig:poisson}
	\end{subfigure}


	\begin{subfigure}[b]{0.45\textwidth}
		\centering
		\includegraphics[width=\textwidth]{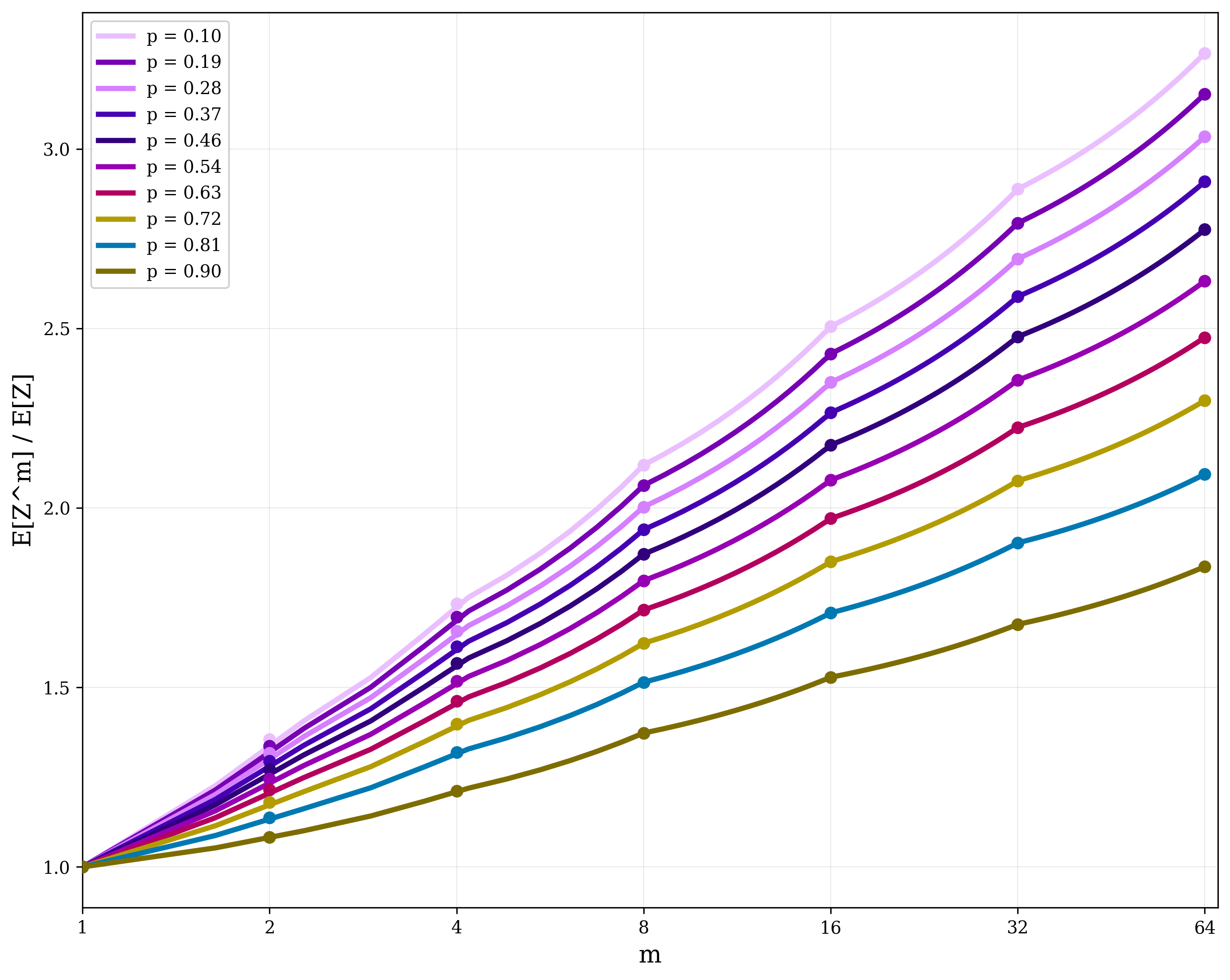}
		\caption{Negative Binomial Distribution}
		\label{fig:negbinom}
	\end{subfigure}
	\begin{subfigure}[b]{0.45\textwidth}
		\centering
		\includegraphics[width=\textwidth]{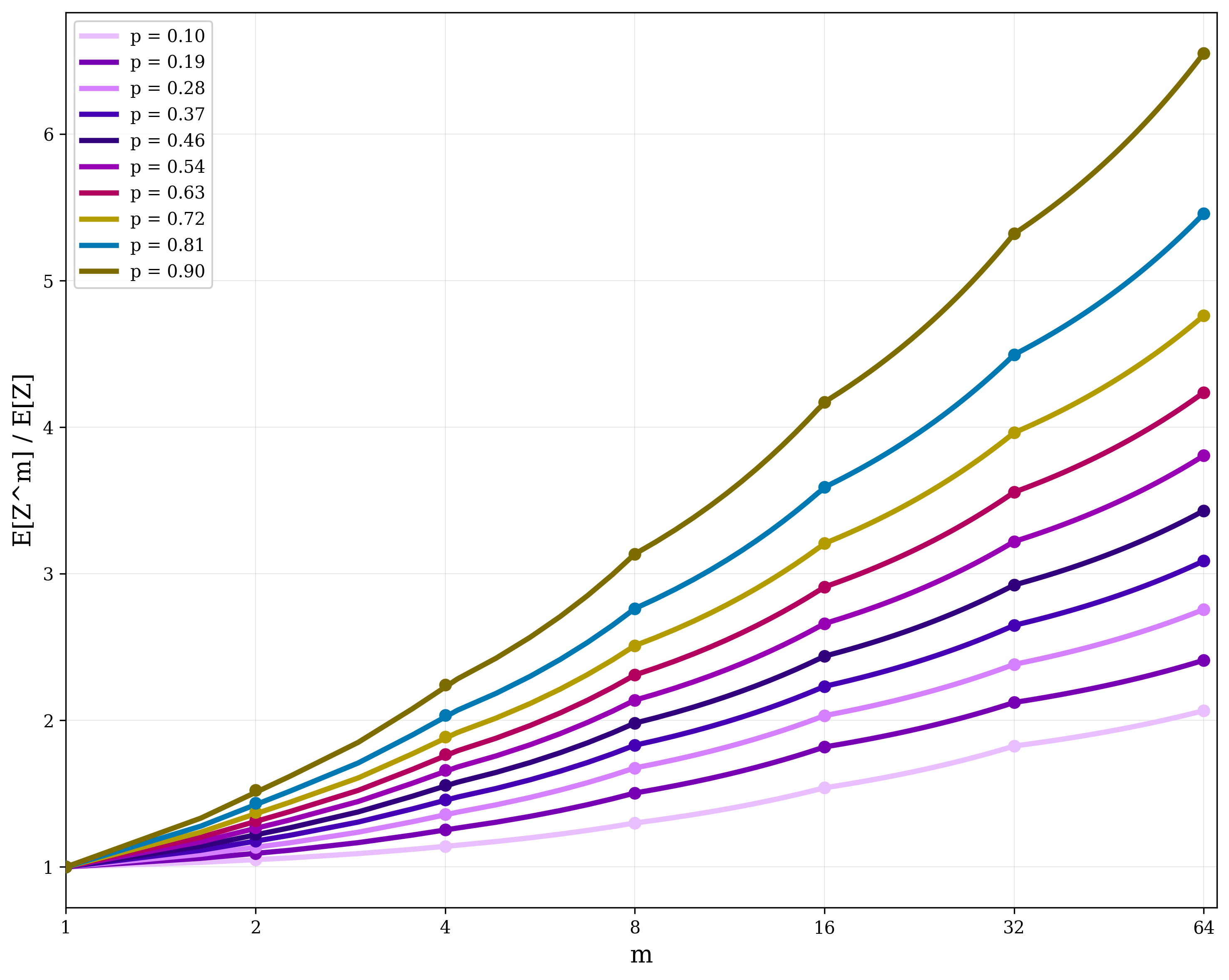}
		\caption{Logarithmic Distribution}
		\label{fig:logarithmic}
	\end{subfigure}
	
	\caption{The expectation ratio vs. the number of dimensions $m$.}
	\label{fig:all_distributions}
\end{figure}


\section{Revenue maximization during congestion}\label{sec:revenue}

In this section, we examine the behavior of the two mechanisms at our second transient-state of interest: the system being \emph{under congestion}.
In such a state, it is natural for \emph{rational} block producers to try to maximize the revenue they are going to obtain, as the number of  transactions eligible for inclusion in the block exceeds the available space. We show that the relevant revenue maximization problem  becomes harder and harder as the number of dimensions grow, while still being relatively easy in one-dimensional setting.

Our main technical idea is reducing the hardness of the revenue maximization problem to that of the multidimensional knapsack problem. 
Some technical complexity arises due to the fact that revenue comes from tips, and thus our reduction must be w.r.t. the tipping strategies adopted by users.

\subsection{Preliminaries}
We start with some preliminary definitions and results.


\paragraph{Revenue Maximization.} The Revenue Maximization (RM) problem asks for a subset of transactions that maximize revenue from tips, subject to resource constraints. The problem is parameterized by: (i) the number of dimensions $m$, (ii) the vector of resource prices $r = (r_1, r_2, \ldots, r_m)$, and (iii) a 
tipping strategy  $f$ that users follow~\footnote{We assume that all parties follow the same strategy, i.e., the best one. Note, that our result holds also when users follow different tipping strategies, as the related problem is as at least hard as the one presented here which is a special case.}. The input consists of 
the vector of resource upper bounds $(G_1, G_2, \ldots, G_m)$, the number of transactions $n$ together with their values and  $m$-dimensional consumption vectors, denoted by $v_j$ and $c^j$, respectively;
the tipping strategy $f$ can depend on $v_j$, $r$, and $c^j$.

The objective is to find a subset $S \subseteq \{1, \ldots, n\}$ of transactions that maximizes the total revenue from tips:
$$
\text{maximize } \sum_{tx_j \in S} f(v_j, r, c^j)
$$
subject to the constraints:
$$\sum_{tx_j \in S} c_i^j \leq G_i, \quad \forall i \in [m]$$

We are going to show that RM is hard for a large class of natural tipping strategies.

\paragraph{Multidimensional Knapsack.} The Multidimensional Knapsack (MDK) problem is parameterized by the number of dimensions $m$, and takes as input  a vector $(G_1, G_2, \ldots, G_m)$ of $m$ bounds, and $n$ items, where each item $item_j$ has an $m$-dimensional vector $c^j $ and an associated value $v_j$. The objective is to find a subset $S \subseteq \{1, \ldots, n\}$ that: 
$$\text{maximizes }\sum_{item_j \in S} v_j$$
subject to the constraints:
$$\sum_{item_j \in S} c_i^j \leq G_i, \quad  \forall i \in [m].$$

Note that, RM is similar to MDK, except of the fact that that the value of each transaction is not determined arbitrarily, but is instead a function that depends on its value, resource consumption, and resource prices.

\paragraph{Hardness of the knapsack problem.}The hardness of MDK increases significantly as the number of dimensions $m$ grow. On the one hand, a polynomial-time approximation scheme (PTAS) is known to exist~\cite{caprara2000approximation} with a runtime of $n^{\theta(m/ε)}$. On the other hand, under the \emph{Gap Exponential Time Hypothesis (Gap-ETH)},  there exist constants $\zeta, \chi, m_0 > 0$ such that for every integer $m > m_0$ and every $\epsilon \in \left(0, \frac{\chi}{\log m}\right)$, there is no $(1 - \epsilon)$-approximation algorithm for $m$-dimensional knapsack that runs in time $O\left(n^{\frac{m}{\epsilon} \cdot \frac{\zeta}{(\log(m/\epsilon))^2}}\right)$.

\begin{theorem}[\cite{doron2024fine}]\label{lemma:MDK}
Under the Gap-ETH assumption, there exist constants $\zeta, \chi, m_0 > 0$ such that for every integer $m > m_0$ and every $\epsilon \in \left(0, \frac{\chi}{\log m}\right)$, there is no $(1 - \epsilon)$-approximation algorithm for $m$-dimensional knapsack that runs in time $O\left(n^{\frac{m}{\epsilon} \cdot \frac{\zeta}{(\log(m/\epsilon))^2}}\right)$.
\end{theorem}

In contrast, for the one-dimensional knapsack problem, the best known PTAS gives us a $(1-\epsilon)$ approximation in time $\tilde{O}(n + 1/\epsilon^2)$ ~\cite{chen2024nearly}.

\subsection{Hardness of RM with monotone tipping}

We are now ready to prove our result.
The class of tipping strategies we are going to consider, is tipping functions $f(v, r, c)$  that are \emph{monotonically increasing} with respect to the transaction value $v$. 
This assumption reflects practical user behavior in transaction fee markets: A higher transaction value $v$ indicates that the transaction is more important or urgent for the user. Consequently, the user is willing to offer a higher tip to incentivize the block producer to include their transaction in the block. By increasing the tip, the user effectively increases the priority of their transaction over others, aligning with rational strategies in competitive fee environments.

Besides monotonicity we are going to need a few extra assumptions on $f$.
Firstly, we are going to assume that the tip is zero when the value of the transaction to the user is zero, and that otherwise the tip is positive. Furthermore, to allow relating the values of items in MDK to tips in RM in the reduction, we restrict our tipping class to functions that are continuous and can be inverted in polynomial time.
We proceed to formally define the  tipping class described above.


\begin{definition}
Let $F_v$ denote the family of tipping functions where  every $f(v,r,c) \in F_v$ is:
\begin{itemize}
	\item  monotonically increasing with respect to the transaction value $v$;
	\item  continuous;
	\item can be inverted in polynomial time;

	\item and where for every $r,c$: $f(v,r,c)>0$, except for $f(0,r,c)$ which is $0$.
\end{itemize}  
\end{definition}

%


We are now ready to prove that RM is as hard as MDK.
\ifFullElse{}{Due to space constrains, the full proof appears in Appendix~\ref{app:proofs}.}

\ignore{
\begin{lemma}\label{lemm:values}
Under the Gap-ETH assumption,  
there is no algorithm that can achieve a $(1 - \epsilon)$-approximation for the RM  problem, parameterized by any $m\in \mathbb{N}, r \in \mathbb{R}^m_{\geq 0}$, and  $f \in F_v$, 
 in time $n^{o\left(m / \epsilon \cdot \frac{1}{(\log (m/\epsilon))^2}\right)}$ for every $\epsilon = O\left(\frac{1}{\log m}\right)$.
\end{lemma}  
}
\begin{theorem}\label{thm:values}
	Under the Gap-ETH assumption, there exist constants $\zeta, \chi, m_0 > 0$ such that for every integer $m > m_0$ and every $\epsilon \in \left(0, \frac{\chi}{\log m}\right)$, there is no $(1 - \epsilon)$-approximation algorithm for the RM problem, parameterized by  $m$, $r \in \mathbb{R}^m_{\geq 0}$, and $f \in F_v$, that runs in time $O\left(n^{\frac{m}{\epsilon} \cdot \frac{\zeta}{(\log(m/\epsilon))^2}}\right)$.
\end{theorem}

\ifFullElse{
	\begin{proof}
	We start by proving one useful preliminary claim regarding 
	the invariance of optimal solutions in MDK when the value of items is scaled by a constant.
	
	\begin{claim}\label{claim:value-scaled}
		Let $S \subseteq \{1, 2, \ldots, n\}$ be the optimal subset of items for a multidimensional knapsack (MDK) problem with $n$ items, where each item $item_j \in [n]$ has a value $v_j$. If we scale the values by a constant $C > 0$, so that each item $item_j$ has value $C \cdot v_j$, then the optimal solution remains the same subset $S$.
	\end{claim}
	\begin{claimproof}
		Consider two MDK instances with the same set of items, identical resource consumption vectors $c^j$ for each item, and identical resource limits $(G_1, G_2, \ldots, G_m)$:
		\begin{enumerate}
			\item Original MDK Instance: Each $item_j$ has value $v_j$.
			\item Scaled MDK Instance: Each $item_j$ has value $C \cdot v_j$, where $C > 0$.
		\end{enumerate}
		
		Let $S \subseteq \{1, 2, \ldots, n\}$ be the optimal solution to the original instance with total value $\sum_{item_j \in S} v_j = \text{OPT}$. We show that $S$ is also optimal for the scaled instance.
		
		For any subset $S' \subseteq \{1, 2, \ldots, n\}$, the objective value in the scaled instance is $\sum_{item_j \in S'} C \cdot v_j =C \cdot \sum_{item_j \in S'} v_j$. Since resource constraints remain identical in both instances, $S$ is feasible in the scaled instance. Suppose, for contradiction, that some subset $S'$ is optimal for the scaled instance instead of $S$. Then:
		$$C \cdot \sum_{item_j \in S'} v_j > C \cdot \sum_{item_j \in S} v_j$$
		
		Dividing by $C$ yields:
		$$\sum_{item_j \in S'} v_j > \sum_{item_j \in S} v_j = \text{OPT}$$
		
		This contradicts the optimality of $S$ in the original instance. Therefore, $S$ remains optimal in the scaled instance.
	\end{claimproof}
	
	For the sake of contradiction, assume there exists some configuration of parameters--specifically, a vector of prices $r$ and a function $f$--such that the RM problem is computationally feasible, meaning
	for all $\zeta,  \chi$, there exist $m$, $\epsilon 
	\in \left(0, \frac{\chi}{\log m}\right)$, $r \in \mathbb{R}^m_{>0}$, $f\in F_v$ such that RM parameterized by $m$, $r$, and $f$ can be approximated to within $(1-ε)$-accuracy in time $O\left(n^{\frac{m}{\epsilon} \cdot \frac{\zeta}{(\log(m/\epsilon))^2}}\right)$. 
	Through a reduction from the MDK problem, we will demonstrate that the RM problem is at least as hard as the MDK problem, implying that for any choice of these parameters, the RM problem remains computationally hard.

	Given an MDK instance with $n$ items, each with resource consumption vector $c^j$, value $v_j$, and resource bounds $(G_1,G_2, \ldots, G_m)$, we construct an equivalent RM instance with same number of transactions as the number of items, i.e., $n'=n$, and with the same number of resource dimensions, i.e., $m'=m$. For each transaction $tx_j$, we set the resource consumption requirements $c'^j_i$ to match those of the corresponding item $item_j$, i.e. $c'^j_i= c_i^j$ for all $i \in [m]$. Additionally, we set the resource limits in RM to be the same as in MDK, i.e., $G_i' = G_i$ for all $i \in [m]$.
	
	Next, we set the transaction value $v'_j$ in RM so that the function $f(v'_j, r, c'^j_i)$ returns exactly $C\cdot v_j$, where $C > 0$ is sufficiently small to ensure that each $C\cdot v_j$ falls within the range $[0,M]$ - there is such $M$ since  the $f$ is positive for positive values $C \cdot v_j$.  Since $f$ is monotone and continuous with respect to the transaction value $v$, we can always select an appropriate $v'_j$ such that $f(v'_j, r, c'^j_i) = C\cdot v_j$.
	Furthermore, $f$ is efficiently invertible, with its runtime being independent of $n$ (as it is not part of the input). Given that for the reduction $m$ is fixed , for large enough $n$ it holds that computing 
	the inverse of $f$ is within $O\left(n^{\frac{m}{\epsilon} \cdot \frac{\zeta}{(\log(m/\epsilon))^2} -1}\right)$.~\footnote{Note, that we need to invert $f$ $n$ times in the reduction.}
	This further implies that the running time of the reduction will also be  within $O\left(n^{\frac{m}{\epsilon} \cdot \frac{\zeta}{(\log(m/\epsilon))^2}}\right)$ as required.

	This construction ensures that the RM and MDK instances have the same feasible region (identical constraints) and proportional objective functions,  and thus the optimal solution of RM can be easily transformed to an optimal solution for MDK. Specifically, for any subset $S \subseteq \{1, 2, \ldots, n\}$, the RM objective value is $$\sum_{tx_j \in S} f(v'_j, r, c'^j) = \sum_{tx_j \in S} C \cdot v_j = C \cdot \sum_{tx_j \in S} v_j$$
	which is exactly $C$ times the MDK objective value for the same subset. By Claim~\ref{claim:value-scaled}, scaling objective values by a positive constant preserves the optimal solution. Therefore, the optimal subset for the RM instance is identical to the optimal subset for the MDK instance, which contradicts Theorem~\ref{lemma:MDK}.
	The theorem follows.
\end{proof}

}{
\begin{proof}[Sketch.]
The proof uses a reduction from the multidimensional knapsack (MDK) problem to show that RM is computationally hard. 
For the sake of contradiction, we first assume that a ``fast'' RM solver exists.
Then, given any MDK instance, we construct an equivalent RM instance with the same resource consumption vectors and bounds, 
and an appropriately scaled objective function. 
More specifically, we set transaction values so that the tip is proportional to the value of the item,
by making use of the monotonicity and continuity of the tipping function.
Combining this with the fact that scaling knapsack objective values by a positive constant keeps the optimal solution the same, 
we show that we can solve MDK instances faster than the time Theorem~\ref{lemma:MDK} allows for.
\end{proof}
}

\begin{remark}
The same reduction as above can be applied for the single-dimension RM problem. Given that as discussed earlier the one-dimensional knapsack problem can be solved in linear time,
it is implied that the single-dimension RM problem can also be solved in linear time, and is thus easy.
\end{remark}

\section{Conclusions and Open Questions}\label{sec:proposals}

Our results in the previous sections indicate that it may challenging to 
use $\md$ in practice. In this section, we describe possible \emph{practical} alternatives
that attempt to balance the steady/transient state performance.
We only comment on the suggested mechanisms at a high-level, and leave a full exploration as an interesting open question.

\paragraph{Slowly-changing gas weights.}
Our first proposal, is for an one-dimensional mechanism with adaptive resource weights $(w_i)_{i\in[m]}$ that are slowly changing based on the level of usage of each resource on chain. We assume the change rate to be a lot slower than the typical stabilization time, e.g., once per day.

On the one hand, the performance of this mechanism in the transient-state remains more or less the same as that of $\oned$.
Revenue maximization is still easy and the mechanism enjoys the
same stability properties as  $\oned$.
On the other hand, by dynamically adjusting the weights of the mechanism in a clever way, we may be able to accommodate higher value traffic, and thus increase the total welfare generated. 

As an example, if higher value traffic consumes a specific resource, we can adjust the weights to increase the its maximum total consumption. The next lemma shows the exact relation of the maximum total consumption and the weights.

\begin{lemma}\label{lemm:max}
	The maximum total consumption of resource $k$ in $\oned$ is:
	$$\min_{i \in [m]}\{w_iG_i\} / w_k, \quad \forall k\in[m]$$
\end{lemma}	

\ifFull{

\begin{proof}
	We denote by $B^1$ the block generated by $\oned$. In the one-dimensional setting, the total gas consumption is bounded by $G$, which  implies that the total consumption of any single resource $k$ is also bounded. To find the maximum consumption that resource $k$ can achieve, we consider the extreme case where $\oned$ allocates all available capacity to resource $k$ only. In this case:
	$$G = \sum_{tx_j \in B^1}\sum_{i \in [m]}w_ic_i^j = w_k \sum_{tx_j \in B^1} c_k^j$$
	
	Since $G$ must satisfy the safety assumption and thus $G \leq w_i G_i$ for all $i \in [m]$, the maximum possible value of $G$ is equal to $\min_{i \in [m]} \{w_iG_i\}$. Therefore, the maximum total consumption of resource $k$ is achieved when $G$ reaches this upper bound:
	$$w_k \sum_{tx_j \in B^1} c_k^j = \min_{i \in [m]} \{w_iG_i\}$$
	
	Solving for the consumption of resource $k$:
	$$\sum_{tx_j \in B^1} c_k^j =\frac{ \min_{i \in [m]} \{w_iG_i\}}{w_k}$$
	
	The lemma follows.
\end{proof}

}

\paragraph{Small number of dimensions.}
Our second proposal, is using only a small number (2-3) of \emph{synthetic} dimensions~\footnote{Interestingly, in concurrent work~\cite{lavee2025does}, limiting the number of dimensions is also explored as a possible solution to a related problem, i.e., that of limited ``resource space'' of $\oned$ compared to $\md$. Here, we also note that this mechanism also performs better than $\md$ when under congestion and w.r.t. stabilization time.}. Each dimension would correspond to a synthetic (gas-like)
resource that (potentially) depends on the $m$ real resources and has a freely determined price based on the usage of the resource.

Our transient-state analysis, shows that for very small $m$ the performance may be at acceptable levels, e.g., solving the related RM problem takes quadratic time on the number of inputs, and the stability time is only lightly affected when $m$ is small. On the other hand, 
increasing the number of dimensions, increases the number of allocations that are possible, and thus provides an opportunity for increased welfare.

\medskip 

Concluding, our analysis suggests that multidimensional pricing may introduce more problems than it solves.
However, slowly-changing parameters or a limited number of dimensions may offer an alternative way of improving the performance of the single dimension mechanism.

\bibliographystyle{splncs04}
\bibliography{tspeed}

\begin{thebibliography}{10}
\providecommand{\url}[1]{\texttt{#1}}
\providecommand{\urlprefix}{URL }
\providecommand{\doi}[1]{https://doi.org/#1}

\bibitem{angeris2024multidimensional}
Angeris, G., Diamandis, T., Moallemi, C.: Multidimensional blockchain fees are
  (essentially) optimal. arXiv preprint arXiv:2402.08661  (2024)

\bibitem{buterinmulti}
Buterin, V.: Multidimensional gas pricing (2024),
  \url{https://vitalik.eth.limo/general/2024/05/09/multidim.html}

\bibitem{buterinethereum}
Buterin, V., Conner, E., Dudley, R., Slipper, M., Norden, I.: Ethereum
  improvement proposal 1559: Fee market change for eth 1.0 chain (2019),
  \url{https://github.com/ethereum/EIPs/blob/master/EIPS/eip-1559.md}

\bibitem{EIP-4844}
Buterin, V., Feist, D., Loerakker, D., Kadianakis, G., Garnett, M., Taiwo, M.,
  Dietrichs, A.: Multidimensional gas pricing (2024),
  \url{https://github.com/ethereum/EIPs/blob/master/EIPS/eip-4844.md}

\bibitem{caprara2000approximation}
Caprara, A., Kellerer, H., Pferschy, U., Pisinger, D.: Approximation algorithms
  for knapsack problems with cardinality constraints. European Journal of
  Operational Research  \textbf{123}(2),  333--345 (2000)

\bibitem{chen2024nearly}
Chen, L., Lian, J., Mao, Y., Zhang, G.: A nearly quadratic-time fptas for
  knapsack. In: Proceedings of the 56th Annual ACM Symposium on Theory of
  Computing. pp. 283--294 (2024)

\bibitem{crapis2024optimal}
Crapis, D., Moallemi, C.C., Wang, S.: Optimal dynamic fees for blockchain
  resources. In: International Conference on Financial Cryptography and Data
  Security. pp. 271--291. Springer (2024)

\bibitem{diamandis2023designing}
Diamandis, T., Evans, A., Chitra, T., Angeris, G.: Designing multidimensional
  blockchain fee markets. In: 5th Conference on Advances in Financial
  Technologies (AFT 2023). pp.~4--1. Schloss Dagstuhl--Leibniz-Zentrum f{\"u}r
  Informatik (2023)

\bibitem{doron2024fine}
Doron-Arad, I., Kulik, A., Manurangsi, P.: Fine grained lower bounds for
  multidimensional knapsack. arXiv preprint arXiv:2407.10146  (2024)

\bibitem{heimbach2025early}
Heimbach, L., Milionis, J.: The early days of the ethereum blob fee market and
  lessons learnt. arXiv preprint arXiv:2502.12966  (2025)

\bibitem{lavee2025does}
Lavee, N., Nisan, N., Pai, M., Resnick, M.: Does your blockchain need
  multidimensional transaction fees? arXiv preprint arXiv:2504.15438  (2025)

\bibitem{EIP-7825}
Rebuffo, G.: Cap on the maximum gas usage of a transaction (2024),
  \url{https://github.com/ethereum/EIPs/blob/master/EIPS/eip-7825.md}

\bibitem{roughgarden2024transaction}
Roughgarden, T.: Transaction fee mechanism design. Journal of the ACM
  \textbf{71}(4),  1--25 (2024)

\end{thebibliography}




\end{document}

\typeout{get arXiv to do 4 passes: Label(s) may have changed. Rerun}